\documentclass[11pt]{article}
\usepackage[top=1in, bottom=1in, left=1in, right=1in]{geometry}
\usepackage[utf8]{inputenc}
\usepackage[T1]{fontenc}

\usepackage[final]{changes}
\usepackage{hyperref}
\usepackage{amsmath,amsthm,amssymb}
\usepackage{mathtools}
\usepackage{color}
\usepackage{graphicx}
\usepackage{enumitem}
\usepackage{xspace}
\usepackage{placeins}
\usepackage{authblk}
\usepackage{wrapfig}
\usepackage[font=small,labelfont=bf,skip=0.5em]{caption}
\usepackage{times}
\usepackage{enumitem}
\setitemize{noitemsep,topsep=0pt,parsep=0pt,partopsep=0pt}
\usepackage{algorithm}
\usepackage[noend]{algpseudocode}

\renewcommand{\paragraph}[1]{\medskip\noindent\textbf{#1}}

\usepackage{caption}

\expandafter\ifx\csname pdfoptionalwaysusepdfpagebox\endcsname\relax\else
\fi

\newif\ifMdBtheorems
\MdBtheoremstrue
\ifMdBtheorems
    \newtheorem{defin}{Definition}
        \newenvironment{definition}{\begin{defin} \sl}{\end{defin}}
    \newtheorem{theo}[defin]{Theorem}
        \newenvironment{theorem}{\begin{theo} \sl}{\end{theo}}
    \newtheorem{lem}[defin]{Lemma}
        \newenvironment{lemma}{\begin{lem} \sl}{\end{lem}}
        
    \newtheorem{propo}[defin]{Proposition}
        
    \newtheorem{coro}[defin]{Corollary}
        \newenvironment{corollary}{\begin{coro} \sl}{\end{coro}}
    \newtheorem{obse}[defin]{Observation}
        \newenvironment{observation}{\begin{obse} \sl}{\end{obse}}
    \newtheorem{rem}[defin]{Remark}
        \newenvironment{remark}{\begin{rem} \rm}{\end{rem}}
    \newtheorem{myfact}[defin]{Fact}
        
\else
    \theoremstyle{plain}
    \newtheorem{theorem}{Theorem}
    \newtheorem{lemma}[theorem]{Lemma}
    
    \newtheorem{corollary}[theorem]{Corollary}

    \theoremstyle{definition}
    \newtheorem{definition}[theorem]{Definition}

    \theoremstyle{remark}
    \newtheorem{remark}[theorem]{Remark}
\fi

\newcommand{\myclaim}[2]
{\begingroup\addtolength\leftmargini{-3mm}\begin{quotation} \noindent {\emph{Claim.} }{#1} \\[2mm]
{\emph{Proof of claim.}} #2 \hfill {\footnotesize $\Box$} \end{quotation}\endgroup}

\newcommand{\myclaimnoqed}[2]
{\begingroup\addtolength\leftmargini{-3mm}\begin{quotation} \noindent {\emph{Claim.} }{#1} \\[2mm]
{\emph{Proof of claim.}} #2 \end{quotation}\endgroup}

\definecolor{ablue}{rgb}{0.3,0.4,0.8}
\definecolor{ared}{rgb}{0.95,0.4,0.4}
\definecolor{agreen}{rgb}{0,0.5,0.25}
\definecolor{ayellow}{rgb}{0.95,0.85,0.3}

\DeclarePairedDelimiter{\oiv}{(}{)}

\renewcommand{\P}{\mathbf{Pr}}

\newcommand{\Reals}{\mathbb{R}}

\newcommand{\cA}{\mathcal{A}}
\newcommand{\cH}{\mathcal{H}}
\newcommand{\cC}{\mathcal{C}}
\newcommand{\cP}{\mathcal{P}}

\newcommand{\cG}{\mathcal{G}}

\newcommand{\cR}{\mathcal{R}}

\newcommand{\cM}{\mathcal{M}}
\newcommand{\cS}{\mathcal{S}}

\newcommand{\sig}{\sigma}
\newcommand{\eps}{\varepsilon}

\newcommand{\eqdef}{\mathrel{\overset{\makebox[0pt]{\mbox{\normalfont\tiny\sffamily def}}}{=}}}

\newcommand{\sd}{\triangle}

\DeclareMathOperator{\poly}{poly}

\DeclareMathOperator{\rdist}{rdist}
\DeclareMathOperator{\opt}{opt}
\DeclareMathOperator{\size}{size}

\DeclareMathOperator{\weight}{weight}

\newcommand{\myin}{\mathrm{in}}
\newcommand{\myout}{\mathrm{out}}
\newcommand{\sigin}{\sig_\myin}
\newcommand{\sigout}{\sig_\myout}
\newcommand{\Pin}{P_\myin}
\newcommand{\Pout}{P_\myout}
\newcommand{\Bin}{B_\myin}
\newcommand{\Bout}{B_\myout}
\newcommand{\Min}{M_\myin}
\newcommand{\Mout}{M_\myout}
\newcommand{\expO}[1]{\exp\!\left(\!O\!\left( #1 \right)\!\right)}
\newcommand{\len}{\mathrm{length}}

\newcommand{\myshort}{\mathrm{short}}
\newcommand{\Sshort}{S_\myshort}
\newcommand{\mymid}{\mathrm{mid}}
\newcommand{\Smid}{S_\mymid}
\newcommand{\mylong}{\mathrm{long}}
\newcommand{\Slong}{S_\mylong}

\newcommand{\mysmall}{\mathrm{small}}
\newcommand{\lsmall}{L_{\mysmall}}
\newcommand{\mymin}{\mathrm{min}}
\newcommand{\imin}{i_{\mymin}}
\newcommand{\mymax}{\mathrm{max}}
\newcommand{\imax}{i_{\mymax}}
\newcommand{\PP}{\mathrm{PP}}

\newcommand{\tspp}{\mbox{{\sc Traveling Salesman Problem}}\xspace}

\newcommand{\TSP}{\mbox{{\sc tsp}}\xspace}
\newcommand{\etsp}{\mbox{{\sc Euclidean TSP}}\xspace}
\newcommand{\mtsp}{\mbox{{\sc Metric TSP}}\xspace}
\newcommand{\btsp}{\mbox{{\sc Euclidean Path Cover}}\xspace}
\newcommand{\NP}{\mbox{{\sc np}}\xspace}
\renewcommand{\P}{\mbox{{\sc p}}\xspace}

\renewcommand{\leq}{\leqslant}
\renewcommand{\geq}{\geqslant}

\newcommand{\etal}{\emph{et~al.}}

\newcommand{\BeginMyItemize}{\begin{itemize}\setlength{\itemsep}{-\parskip}}
\newcommand{\EndMyItemize}{\end{itemize}}

\newcommand{\BeginMyEnumerate}{\begin{enumerate}\setlength{\itemsep}{-\parskip}}
\newcommand{\EndMyEnumerate}{\end{enumerate}}

\newcommand{\defproblem}[3]{
\begin{quotation}
\noindent
  \textsc{#1}  \\
  {\bf{Input:}} #2 \\
  {\bf{Question:}} #3
\end{quotation}
}

\author[1]{Mark de Berg}
\author[2]{Hans L. Bodlaender}
\author[3]{S\'andor Kisfaludi-Bak}
\author[4]{Sudeshna Kolay}

\affil[1]{Department of Mathematics and Computer Science, Eindhoven University of Technology, Eindhoven, The Netherlands\authorcr\texttt{m.t.d.berg@tue.nl}}
\affil[2]{Department of Computer Science, Utrecht University, Utrecht, The Netherlands\authorcr\texttt{h.l.bodlaender@uu.nl}}
\affil[3]{Department of Computer Science, Aalto University, Espoo, Finland\authorcr\texttt{sandor.kisfaludi-bak@aalto.fi}}
\affil[4]{Indian Institute of Technology Kharagpur, Kharagpur, India\authorcr\texttt{
skolay@cse.iitkgp.ac.in}}

\title{An ETH-Tight Exact Algorithm for Euclidean TSP\thanks{This work was supported by the NETWORKS project, funded by the Netherlands Organization for Scientific Research NWO under project no. 024.002.003.}}

\begin{document}
\DeclareRobustCommand{\BIBSORT}[2]{#2} 
\nocite{doktori,BergBKMZ20} 

\maketitle

\begin{abstract}
We study exact algorithms for \etsp in $\Reals^d$. In the early 1990s algorithms with
$n^{O(\sqrt{n})}$ running time were presented for the planar case, and some years later
an algorithm with $n^{O(n^{1-1/d})}$ running time was presented for any $d\geq 2$.
Despite significant interest in subexponential exact algorithms over the past decade,
there has been no progress on \etsp, except for a lower
bound stating that the problem admits no $2^{o(n^{1-1/d})}$ algorithm unless
ETH fails. In this paper we settle the complexity of \etsp, up to constant factors in the exponent and under ETH, by giving an algorithm with running time $2^{O(n^{1-1/d})}$.
\end{abstract}

\section{Introduction}
The \tspp, or \TSP for short, is one of the most widely studied problems in all of
computer science. In (the symmetric version of) the problem we are given a complete
undirected graph~$\cG$ with positive edge weights, and the goal is to compute a
minimum-weight cycle visiting every vertex exactly once.
In 1972 the problem was shown to be \NP-hard by Karp~\cite{Karp10}.
A brute-force algorithm for \TSP runs in $O(n!)$, but the celebrated
Held-Karp dynamic-programming algorithm, discovered independently by
Held and Karp~\cite{HeldK61} and Bellman~\cite{Bellman-TSP}, runs in $O(2^n n^2)$ time.
Despite extensive efforts and progress on special cases,
it is still open if  an exact algorithm for
\TSP exists with running time $O((2-\eps)^n \poly(n))$.

In this paper we study the Euclidean version of \TSP, where the input is a set
$P$ of $n$ points in~$\Reals^d$ and the goal is to find a tour of minimum
Euclidean length visiting all the points. \etsp has been studied extensively
and it can be considered one of the most important geometric optimization problems.
Already in the mid-1970s, \etsp was shown to be \NP-hard~\cite{GareyGJ76,Papadimitriou77}.
Nevertheless, its computational complexity is markedly different from that of the
general \TSP problem. For instance, \etsp admits efficient approximation algorithms.
Indeed, the famous algorithm by Christofides~\cite{Chr76}---which actually works
for the more general \mtsp problem---provides a $(3/2)$-approximation
in polynomial time, while no polynomial-time approximation algorithm exists
for the general problem (unless \P=\NP). It was a long-standing open problem
whether \etsp admits a PTAS. The question was answered affirmatively by
Arora~\cite{Arora98} who provided a PTAS with running time~$n (\log n)^{O(\sqrt{d}/\eps)^{d-1}}$. Independently, Mitchell~\cite{Mitchell99} designed a PTAS in $\mathbb{R}^2$. The running time was
improved to $2^{(1/\eps)^{O(d)}} n + (1/\eps)^{O(d)} n\log n$ by Rao and Smith~~\cite{RaoS98},
\added{and then to $2^{(1/\eps)^{O(d)}}n$ by
Bartal and Gottlieb~\cite{DBLP:conf/focs/BartalG13}.
Kisfaludi-Bak~\etal~\cite{DBLP:conf/focs/Kisfaludi-BakNW21} recently developed a PTAS
with running time $2^{O(1/\eps^{d-1})} n\log n$, and they showed that the dependency
on $\eps$ is optimal under Gap-ETH.}
Hence, the computational complexity of approximating \etsp is well-understood.

Results on exact algorithms for \etsp---such algorithms are the topic of our paper---are also
quite different from those on the general problem. The best known algorithm for the general
case runs, as already remarked, in exponential time, and there is no $2^{o(n)}$ algorithm under the Exponential Time Hypothesis
(ETH)~\cite{ImpagliazzoP01} due to classical reductions for {\sc Hamiltonian Cycle}~\cite[Theorem 14.6]{CyganFKLMPPS15}.
\etsp, on the other hand, is solvable in subexponential time.
For the planar case this has been shown in the early 1990s by
Kann~\cite{Kann92} and independently by  Hwang, Chang and Lee~\cite{HwangCL93}, who presented an algorithm with an $n^{O(\sqrt{n})}$ running time.
Both algorithms use a divide-and-conquer approach that relies on finding a suitable
separator. The approach taken by Hwang, Chang and Lee
is based on considering a triangulation of the point set
such that all segments of the tour appear in the triangulation, and then observing that the
resulting planar graph has a separator of size $O(\sqrt{n})$.
Such a separator can be guessed in $n^{O(\sqrt{n})}$ ways, leading to a recursive algorithm with
 $n^{O(\sqrt{n})}$ running time. It seems hard to extend this
approach to higher dimensions. Kann obtains his separator in a more geometric
way, using the fact that in an optimal tour, there cannot be too many long edges
that are relatively close together---see the Packing Property we formulate
in Section~\ref{sec:bal-sep}. This makes it possible to compute a separator
that is crossed by $O(\sqrt{n})$ edges of an optimal tour, which can
be guessed in $n^{O(\sqrt{n})}$ ways. The geometric flavor of this algorithm
makes it more amenable to extensions to higher dimensions. Indeed,
some years later Smith and Wormald~\cite{SmithW98} gave an algorithm
for \etsp in $\Reals^d$, which is based on a similar kind of geometric
separator as used by Kann. Their algorithm runs in $n^{O(n^{1-1/d})}$ time. (Here and in the sequel we consider the dimension~$d$ to be an absolute constant.)

The main question, also posed by Woeginger in his survey~\cite{Woeginger08}
on open problems around exact algorithms, is the following: is an exact algorithm with running time $2^{O(n^{1-1/d})}$ attainable for \etsp? Similar results have been
obtained for some related problems. In particular, Deineko~\etal~\cite{DeinekoKW06}
proved that Hamiltonian Cycle on planar graphs can be solved in $2^{O(\sqrt{n})}$
time, and Dorn~\etal~\cite{DornPBF10}
proved that \TSP on weighted planar graphs can be solved in $2^{O(\sqrt{n})}$ time.
Marx and Sidiropoulos~\cite{MarxS14} have shown that \etsp does not admit
an algorithm with $2^{O(n^{1-1/d-\eps})}$, unless ETH fails. Recently this conditional lower bound was strengthened to $2^{\Omega(n^{1-1/d})}$ by De Berg~\etal~\cite{BergBKMZ20}.
In the past twenty years the algorithms for \etsp have not been improved, however.
Hence, even for the planar case the complexity of \etsp is still unknown.

\paragraph{Our contribution}
We settle the complexity of \etsp, up to constant factors in the exponent:  we
present an algorithm for \etsp in $\Reals^d$ with running time $2^{O(n^{1-1/d})}$, where $d\geq 2$ is a fixed
constant. \replaced{This is best possible}{Our algorithm is asymptotically optimal}, as De Berg~\etal\cite{BergBKMZ20} show
that there is no $2^{o(n^{1-1/d})}$ algorithm unless ETH fails.

The global approach to obtain the upper bound is similar to the approach of Kann~\cite{Kann92}
and Smith and Wormald~\cite{SmithW98}: we use a divide-and-conquer algorithm
based on a geometric separator. A geometric separator for a given point set~$P$
is a simple geometric object---we use a hypercube---such that the number of points
inside the separator and the number of points outside the separator are roughly balanced.
As mentioned above, Kann~\cite{Kann92} and Smith and Wormald~\cite{SmithW98} use
a Packing Property of the edges in an optimal \TSP tour to argue that a separator
exists that is crossed by only $O(n^{1-1/d})$ edges from the tour. Since $P$
defines $\binom{n}{2}$ possible edges, the set of crossing edges can be guessed
in $n^{O(n^{1-1/d})}$ ways.

The first obstacle we must overcome if we want to
beat this running time is therefore that the number of subproblems is already too
large at the first step of the recursive algorithm. Unfortunately there is no
hope of obtaining a balanced separator that is crossed by $o(n^{1-1/d})$
edges from the tour: there are point sets such that
any balanced separator that has a ``simple'' shape (e.g., ball or hypercube) is crossed $\Omega(n^{1-1/d})$ times by an optimal
tour. Thus we proceed differently: we prove that there exists a separator
such that, even though it can be crossed by up to $\Theta(n^{1-1/d})$ edges
from an optimal tour, the total number of candidate subsets of crossing edges
we need to consider is only $2^{O(n^{1-1/d})}$. We obtain such a separator
in two steps. First we prove a \emph{distance-based separator theorem}
for point sets. Intuitively, this theorem states that any
point set~$P$ admits a balanced separator such that the number
of points from~$P$ within a certain distance from the separator decreases rapidly
as the distance decreases. \added{This is useful because an optimal tour
cannot have many ``long'' edges crossing the separator due to the Packing Property.
Hence, limiting the number of points ``close'' to the separator also limits
the number of combinations of edges that form a candidate subset.}
In the second step we then prove that this
separator~$\sig$ \added{indeed} has the required properties, namely
(i) $\sig$ is crossed by $O(n^{1-1/d})$ edges in an optimal tour, and
(ii) the number of candidate sets of crossing edges is $2^{O(n^{1-1/d})}$.
In order to prove these properties we use the Packing Property of the
edges in an optimal tour.

There is one other obstacle we need to overcome to obtain a $2^{O(n^{1-1/d})}$
algorithm: after computing a suitable separator~$\sig$ and guessing a set~$S$ of
crossing edges, we still need to solve many different subproblems.
The reason is that the partial solutions on either side of~$\sig$ need
to fit together into a tour on the whole point set. Thus a
partial solution on the outside of~$\sig$ imposes connectivity constraints
on the inside. More precisely, if $B$ is the set of endpoints of the
edges in~$S$ that lie inside $\sig$, then the subproblem we face
inside $\sig$ is as follows: compute a set of paths visiting the points
inside~$\sig$ such that the paths realize a given matching on~$B$.
The number of matchings on $|B|$ boundary points is $|B|^{\Theta(|B|)}$,
which is again too much for our purposes.
Fortunately, the rank-based approach~\cite{single-exponential,CyganKN18}
developed in recent years can be applied here. By applying this
approach in a suitable manner, we then obtain our $2^{O(n^{1-1/d})}$
algorithm.

\paragraph{A word on the model of computation}
In this paper we are mainly interested in the combinatorial complexity of \etsp.
The algorithm we describe in Sections~\ref{sec:bal-sep} and \ref{sec:alg}  therefore works in the
real-RAM model of computation, with the capability of taking square roots.
In particular, we assume that distances can be added in $O(1)$ time, so that the length
of a given tour can be computed exactly in $O(n)$ time.
In Section~\ref{sec:packing} we also consider the following ``almost Euclidean'' version of
the problem: we are given a set $P=\{p_1,\ldots,p_n\}$ with rational coordinates,
together with a distance matrix~$D$ such that $D[i,j]$ contains an approximation
of $|p_i p_j|$. The property we require is that the ordering of distances is
preserved: if $|p_i p_j| < |p_k p_l|$ then $D[i,j] < D[k,l]$.
We show that an optimal tour in this setting satisfies the Packing Property,
which implies that our algorithm can solve the almost Euclidean version
of \etsp in $2^{O(n^{1-1/d})}$ time.

\section{A separator theorem for TSP}\label{sec:bal-sep}
In this section we show how to obtain a separator that can be used as the basis
of an efficient recursive algorithm to compute an optimal \TSP tour.
Intuitively, we need a separator that is crossed few times by an
optimal solution and such that the number of candidate sets of crossing edges is small.
\replaced{To this end we construct a separator $\sig$ such that there are few points
relatively close to $\sig$. Using the Packing Property we then show
that~$\sig$ has all the desired properties.}{We obtain such a separator in two steps:
first we construct a separator~$\sig$ such that
there are only few points relatively close to~$\sig$, and then we show that this implies
that~$\sig$ has all the desired properties.}

\paragraph{Notation and terminology}
We define a \emph{separator} to be
the boundary of an axis-aligned hypercube. A separator~$\sig$ partitions
$\Reals^d$ into two regions: a region~$\sigin$ consisting of all points in $\Reals^d$ inside
or on~$\sig$, and a region~$\sigout$ consisting of all points in $\Reals^d$ strictly outside~$\sig$.
We define the \emph{size} of a separator~$\sig$ to be its side length, and we denote it
by~$\size(\sig)$. For a separator~$\sig$ and a scaling factor $t\geq 0$, we define
$t\sig$ to be the separator obtained by scaling $\sig$ by a factor $t$ with respect
to its center. More precisely, $t\sig$ is the separator whose center is the same as
the center of~$\sig$ and with $\size(t\sig) = t\cdot \size(\sig)$;
see Fig.~\ref{fig:defs}(i).
\begin{figure}
\begin{center}
\includegraphics{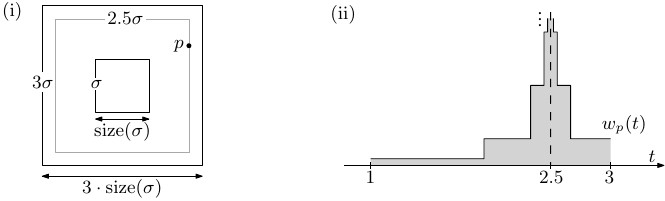}
\end{center}
\caption{(i) A separator $\sig$ and a point $p$ with $\rdist(p,\sig)=0.75$.
          (ii) Schematic drawing of the weight function $w_p(t)$
               of a point $p$ such that $\rdist(p,t\sig^*)=0$ for $t=2.5$.}
\label{fig:defs}
\end{figure}

Let $P$ be a set of $n$ points in $\Reals^d$.
A separator~$\sig$ induces a partition of the given point set~$P$ into two subsets,
$P\cap \sigin$ and $P\cap\sigout$. We are interested in \emph{$\delta$-balanced separators},
which are separators such that $\max(|P\cap \sigin|,|P\cap \sigout|) \leq \delta n$ for a
fixed constant $\delta>0$. It will be convenient to work with
$\delta$-balanced separators for $\delta=4^d/(4^d+1)$. From now on we
will refer to $(4^d/(4^d+1))$-balanced separators simply as \emph{balanced separators}.
(There is nothing special about the constant $4^d/(4^d+1)$, and it could be
made smaller by a more careful reasoning and at the cost of some other constants we will encounter later on.)

\paragraph{Distance-based separators for point sets}
As mentioned, we first construct a separator~$\sig$ such that there are only few points
close to it. To this end we define the \emph{relative distance} from a point~$p$ to~$\sig$, denoted by $\rdist(p,\sig)$, as follows:
\[
\rdist(p,\sig) \eqdef d_{\infty}(p,\sig) / \size(\sig),
\]
where $d_{\infty}(p,\sig)$ denotes the shortest distance in the $\ell_{\infty}$-metric
between $p$ and any point on~$\sig$. Note that if $t$ is the scaling
factor such that $p \in t\sig$, then $\rdist(p,\sig) = |1-t|/2$.
For integers~$i$ define
\[
P_i(\sig) \eqdef \{ \ p\in P \ \mid \ \rdist(p,\sig)  \leq 2^i / n^{1/d} \ \}.
\]
Note that the smaller $i$ is, the closer to $\sig$ the points in $P_i(\sig)$ are required to be.
We now wish to find a separator $\sig$ such that the size of the sets $P_i(\sig)$ decreases rapidly as $i$ decreases.
\begin{theorem}\label{thm:sep-point-distance}
Let $P$ be a set of $n$ points in $\Reals^d$. Then there is a balanced separator~$\sig$ for $P$ such that
\[
|P_i(\sig)| \ = \
\left\{ \begin{array}{ll}
        O((3/2)^i n^{1-1/d}) & \mbox{for all $i < 0$} \\[2mm]
        O(4^i n^{1-1/d}) & \mbox{for all $0 \leq i$}
        \end{array}
\right.
\]
Moreover, such a separator can be found in $O(n^{d+1})$ time.
\end{theorem}
\begin{proof}
Let $\sig^*$ be a smallest separator such
that $|P\cap \sigin^*| \geq n/(4^d+1)$. We will show that there is a $t^*$
with $1\leq t^*\leq 3$ such that $t^* \sig^*$ is a separator with the
required properties.
\medskip

First we claim that $t \sig^*$ is balanced for all $1\leq t\leq 3$. To see
this, observe that for $t\geq 1$ we have
\[
|P\cap (t\sig^*)_\myout|
    \leq |P\cap \sigout^*|
    = n - |P\cap \sigin^*|
    \leq n-n/(4^d+1)
    = (4^d/(4^d+1)) n.
\]

Moreover, for $t\leq 3$ we can cover $t\sigin^*$ by $4^d$ hypercubes of size
$(3/4)\cdot\size(\sig^*)$. By definition of $\sig^*$ these hypercubes
contain less than $n/(4^d+1)$ points each, so
$
|P\cap (t\sig^*)_\myin| < 4^d \cdot (n/(4^d+1)),
$
which finishes the proof of the claim.
\added{(So far the proof is similar to earlier separator constructions~\cite{BergBKMZ20,DBLP:journals/siamcomp/Har-PeledQ17};
the main challenge is to establish the distance properties, which we do next.)}
\medskip

It remains to  prove that there is a $t^*$ with $1\leq t^*\leq 3$ such that
$t^*\sig^*$ satisfies the condition on the sizes of the sets $P_i(t^*\sig^*)$.
To this end we will define a weight function~$w_p:[1,3]\rightarrow \Reals$ for
each $p\in P$. The idea is that the closer $p$ is to $t\sig^*$, the higher the
value~$w_p(t)$. An averaging argument will then show that there must be a
$t^*$ such that $\sum_{p\in P} w_p(t^*)$ is sufficiently small, from which it
follows that $t^*\sig^*$ satisfies the condition on the sizes of the sets
$P_i(t^*\sig^*)$. Next we make this idea precise.

Assume without loss of generality that $\size(\sigma^*)=1$. For a point $p\in
P$, let $i_p(t)$ be the integer such that  $2^{i_p(t)-1}/n^{1/d} <
\rdist(p,t\sig^*) \leq 2^{i_p(t)}/n^{1/d}$, where  $i_p(t)=-\infty$ if
$\rdist(p,t\sig^*)=0$. Note that $p\in P_i(t\sig^*)$ if and only if $i_p(t)
\leq i$. We define the weight function $w_p(t)$ as follows; see
Fig.~\ref{fig:defs}(ii).

\[
w_p(t) \ \eqdef \
\left\{ \begin{array}{ll}
        \frac{n^{1/d}}{(3/2)^{i_p(t)}} & \mbox{if $i_p(t) < 0$ } \\[2mm]
        \frac{n^{1/d}}{4^{i_p(t)}} & \mbox{if $i_p(t) \geq 0$} \\[2mm]
        \mbox{undefined} & \mbox{if $i_p(t) = -\infty$.}\\
       \end{array}
\right.
\]

\added{The above definition has different constants in the denominators ($(3/2)^{i_p(t)}$ versus $4^{i_p(t)}$) for a technical reason: we will soon need that when multiplied by $2^{i_p(t)}$, the cumulative sums are convergent sequences for both $i_p(t)<0$ and $i_p(t)\geq 0$. One could instead use any constant from the open interval $(1,2)$ for the case $i_p(t)<0$, and any constant from $(2,\infty)$ for the case $i_p(t)\geq 0$, respectively.}
We now want to bound $\int_1^3 w_p(t) \mathrm{d}t$. Note that the function
$w_p(t)$ may be undefined for at most one $t\in [1,3]$, namely when there is a
$t$ in this range such that $\rdist(p,t\sig^*)=0$.  Formally we should remove
such a~$t$ from the domain of integration. To avoid cluttering the notation we
ignore this technicality and continue to write $\int_1^3 w_p(t) \mathrm{d}t$.

\myclaimnoqed{
For each $p \in P$, we have $\int_1^3 w_p(t) \mathrm{d}t = O(1)$.
}{
Define  $T_p(i) \eqdef \{ t \mid 1\leq t\leq 3 \mbox{ and } i_p(t) = i \}$. By the
definition of~$w_p(t)$,  the value $w_p(t)$ is constant over $T_p(i)$. We
therefore want to bound $\mathrm{len}(T_p(i))$, the sum of the lengths of the
intervals comprising~$T_p(i)$.  Assume without loss of generality that the
center of $\sig^*$ lies at the origin  of~$\mathbb{R}^d$. Then, depending on
the position of $p$, there is some coordinate $x_i$ such that
$\rdist(p,t\sig^*) = \vert p_{x_i}-t\vert /
\size(t\sigma^*)$.
Assume without loss of generality that $\vert p_{x_1}-t \vert=\size(t\sigma^*)$.
Since $1\leq t\leq 3$ and $\size(\sigma^*)=1$, we then have
$\rdist(p,t\sig^*) \geq \vert p_{x_1}-t\vert/3$.  Hence, for any $t\in T_p(i)$ we
have $\vert p_{x_1}-t\vert/3 \leq 2^i / n^{1/d}$.  This implies that
$\mathrm{len}(T_p(i)) \leq 6 \cdot 2^i/n^{1/d}$ and so

\begin{align*}
\hspace{2cm}\int_1^3 \! w_p(t) \mathrm{d}t
&= \sum\limits_{i\geq 0} \mathrm{len}(T_p(i)) \cdot \frac{n^{1/d}}{4^i} + \sum\limits_{i < 0} \mathrm{len}(T_p(i)) \cdot \frac{n^{1/d}}{(3/2)^i} \\
&\leq 6 \cdot \sum\limits_{i\geq 0} \left(\frac{1}{2}\right)^{\!i} \!+\!  6\cdot \sum\limits_{i < 0} \left(\frac{4}{3}\right)^{\!i}\\
&= 30. \qedhere
\end{align*}
}

The above claim implies that $\int_1^3 \left( \sum_{p\in P} w_p(t) \right) \mathrm{d}t \leq 30n$.
Hence there exists a $t^*\in [1,3]$ such that $\sum_{p\in P} w_p(t^*) \leq 15n$.
Now consider a set $P_i(t^*\sig^*)$ with
$i\geq 0$. Each $p\in P_i(t^*\sig^*)$ has
$i_p(t^*) \leq i$ and so
\[
|P_i(t^*\sig^*)|
    \leq \frac{\sum_{p\in P} w_p(t^*)}{\min_{p\in P_i(t^*\sig^*)} w_p(t^*) }
    = \frac{O(n)}{n^{1/d}/4^i }
    = O(4^i n^{1-1/d}).
\]
A similar argument shows that
$|P_i(t^*\sig^*)| = O((3/2)^i n^{1-1/d})$ for all $i < 0$.

To find the desired separator we first compute~$\sig^*$. Note that we can
always shift $\sig^*$ such that it has at least one point on at least $d$ of
its $(d-1)$-dimensional faces. (Note that an input point on a
lower-dimensional face of $\sig^*$ is counted towards each incident facet.)
Hence, a simple brute-force algorithm can find $\sig^*$ in $O(n^{d+1})$ time.
Once we have $\sig^*$, we would like to find the value~$t^*\in[1,3]$
minimizing $\sum_{p\in P} w_p(t)$. Recall that each~$w_p$ is a step function,
and so $\sum_{p\in P} w_p$ is a step function as well. There is one slight
issue, however, namely that the number of steps of the functions~$w_p$ is
unbounded. We deal with this issue by replacing each~$w_p$ by a truncated
version $\overline{w}_p$, as explained next.

We define the truncated function $\overline{w}_p$ as
follows:

\[
\overline{w}_p(t) \ \eqdef \
\left\{ \begin{array}{ll}
       1/n & \mbox{if $w_p(t) < 1/n$} \\[2mm]
       w_p(t) & \mbox{if $1/n \leq w_p(t) \leq 15 n$} \\[2mm]
       15 n & \mbox{if $w_p(t) > 15 n$. }
      \end{array}
\right.
\]

Each function $\overline{w}_p$ is a step function, and one easily verifies
that $\overline{w}_p$ has $O(\log n)$ steps which we can compute in $O(\log
n)$ time.  Hence, we can find a value $\overline{t}$ that minimizes
$\sum_{p\in P} \overline{w}_p(t)$ in $O(n \log n)$ time, \added{by scanning over
all step functions in parallel and maintaining their sum as we go.} Since $\sum_{p\in P}
w_p(\overline{t})=O(n)$ if $\sum_{p\in P} \overline{w}_p(\overline{t})=O(n)$,
the separator $\overline{t}\sig^*$ has the required properties.
\end{proof}


\begin{remark}
It is not hard to reduce the time needed to compute
the separator by working with an approximation of the smallest
hypercube~$\sig^*$ containing at least $n/(4^d+1)$ points.
We can find an $\eps$-approximation to the minimum enclosing
ball in linear time~\cite{Har-Peled:2015:NPL:2856350.2831230}, whose circumscribed
axis-parallel cube is a constant-approximation to the minimum enclosing cube.
Note that this would weaken the balance factor of the separator theorem.
Nonetheless, in our application this does not make a difference,
and the simple brute-force algorithm to find~$\sig^*$ suffices.
\end{remark}

In the remainder we will need a slightly more general version of
Theorem~\ref{thm:sep-point-distance}, where we require the separator to be balanced
with respect to a given subset $Q\subseteq P$, that is, we require
$\max(|Q\cap \sigin|,|Q\cap \sigout|) \leq \delta |Q|$ for $\delta=4^d/(4^d+1)$.
Note that the distance condition in the corollary below is still with respect to the points
in~$P$. The proof of the corollary is exactly the same as before, we only need
to redefine $\sig^*$ to be a smallest separator such
that $|Q\cap \sigin^*| \geq |Q|/(4^d+1)$.
\begin{corollary}\label{coro:sep-point-distance}
Let $P$ be a set of $n$ points in $\Reals^d$ and let $Q\subseteq P$. Then there is a separator~$\sig$
that is balanced with respect to $Q$ and such that

\[
|P_i(\sig)| \ = \
\left\{ \begin{array}{ll}
        O((3/2)^i n^{1-1/d}) & \mbox{for all $ i < 0$} \\[2mm]
        O(4^i n^{1-1/d}) & \mbox{for all $0 \leq i$.}
        \end{array}
\right.
\]
Moreover, such a separator can be found in $O(n^{d+1})$ time.
\end{corollary}

\paragraph{A separator for TSP}
Let $P$ be a set of $n$ points in~$\Reals^d$, and let $\cS(P)$ be the set of segments
defined by~$P$, that is, $\cS(P) \eqdef \{ pq \mid (p,q)\in P\times P\}$.
Now consider a segment $s\in \cS(P)$ and a separator~$\sig$. We say that
$s$~\emph{crosses}~$\sig$ if one endpoint of $s$ lies in $\sigin$ while
the other lies in~$\sigout$. Using our distance-based separator for points we
want to find a separator that is crossed only a few times by an optimal TSP tour.
Moreover, we want to control the number of ways in which we have to ``guess''
a set of crossing segments. For this we will need the following
crucial property of the segments in an optimal TSP tour.
\begin{definition}\label{def:packing}
A set $S$ of segments in $\Reals^d$ has the \emph{Packing Property} if
for \replaced{every axis-parallel hypercube~$\sigma$}{any separator~$\sig$} \replaced{we have both of the following properties.}{we have}
\begin{itemize}
\item (PP1):
	$\big| \{ s\in S \mid s \mbox{ crosses } \sig \mbox{ and } \len(s) \geq \size(\sig) \}  \big| \!=\! O(1)$ \\[-4mm]
\item (PP2):
	$\big|  \{ s\in S \mid s \subset \sigin \mbox{ and }  \len(s) \geq \size(\sig)/4 \}  \big| \!=\! O(1)$.
\end{itemize}
\end{definition}
Property~(PP2) is actually implied by (PP1), but it will
be convenient to explicitly state (PP2) as part of the definition.
Note that the constants hidden in the $O$-notation in Definition~\ref{def:packing}
may (and do) depend on~$d$.

Some variants of the above Packing Property have been shown to hold for the
set of edges of an optimal tour for \etsp~\cite{Kann92,SmithW98}. For completeness we
present a proof in a more general setting, which can be found in
Section~\ref{sec:packing}.\footnote{\added{In this more general setting, we work with
a distance matrix~$D$ instead of with the exact Euclidean distances. We prove that
the Packing Property holds when the pairwise distances in~$D$ have the same ordering
as the pairwise Euclidean distances, that is, when $D[i,j]<D[k,\ell]$ if and only if $|p_i p_j| < |p_k p_{\ell}|$.}}
Hence, we can restrict our attention to subsets
of~$\cS(P)$ with the Packing Property. For a separator~$\sig$, we are thus
interested in the following collection of sets of segments crossing~$\sig$:
\[
\cC(\sig,P) \eqdef
  \{  S \subseteq \cS(P) \mid \mbox{ $S$ has the Packing Property and all segments in $S$ cross $\sig$} \}.
\]
Our main separator theorem\added{, given as Theorem~\ref{thm:sep-tsp} below,}
states that we can find a separator~$\sig$ that
is balanced and such that the sets in $\cC(\sig,P)$, as well as
the collection $\cC(\sig,P)$ itself, are small.
\replaced{Given a separator~$\sig$, it is hard to generate the set $\cC(\sig,P)$,
because even for a given set~$S$ it is not easy to
check that $S$ has the Packing Property. We will therefore generate
a slightly larger collection of \emph{candidate sets}, which we denote by $\cC'(\sig,P)$.}{Since the general packing
property is hard to test, we will enumerate a slightly
larger collection of \emph{candidate sets}, which we denote by $\cC'(\sig,P)$.}
\index{candidate set}

\begin{theorem}\label{thm:sep-tsp}
Let $P$ be a set of $n$ points in $\Reals^d$ and let $Q\subseteq P$.
Then there is a separator~$\sig$ and collection $\cC'(\sig,P)$ such that
\begin{enumerate}
\item[(i)] $\sig$ is balanced with respect to $Q$
\item[(ii)] each candidate set $S\in \cC'(\sig,P)$ contains $O(n^{1-1/d})$ segments
\item[(iii)] $\cC(\sig,P) \subseteq \cC'(\sig,P)$, and $|\cC(\sig,P)| \leq |\cC'(\sig,P)|= 2^{O(n^{1-1/d})}$.
\end{enumerate}
Moreover, $\sig$ and the collection $\cC'(\sig,P)$ can be computed in $2^{O(n^{1-1/d})}$ time.
\end{theorem}
\begin{proof}
Let $\sig$ be the separator obtained by applying Corollary~\ref{coro:sep-point-distance}
to the sets $P$ and $Q$. Then $\sig$ has property~(i). It remains to show that there is a set $\cC'(\sig,P)$ that has properties (ii) and (iii) and that can be enumerated in $2^{O(n^{1-1/d})}$ time. We assume without loss of generality that $\size(\sig)=1$
and that $\sig$ is centered at the origin.

Let $\lsmall \eqdef 1/(n^{1/d} n^{(1-1/d)\log_{3/2} 2})$.
Any set $S\in \cC(\sig,P)$ can be partitioned into three subsets:
\begin{itemize}
\item $\Sshort \eqdef \{ s\in S \mid \len(s) \leq  \lsmall \}$
\item $\Smid \eqdef \{ s\in S \mid \lsmall < \len(s) \leq 1 \}$
\item $\Slong \eqdef \{ s\in S \mid \len(s) > 1 \}$.
\end{itemize}
We will analyze each of the three subsets separately. This analysis will give an upper bound on the number of candidates for the given type of subset. The argument will be constructive in the sense that it allows us to enumerate the candidates efficiently. By combining the candidates for the three types of subsets, we then get the desired set $\cC'(\sig,P)$. We start by analyzing $\Sshort$ and $\Slong$.

\myclaim{
For any $S\in \cC(\sig,P)$ the set $\Sshort$ consists of $O(1)$ segments, and
the number of different subsets $\Sshort$ that can arise over all sets $S\in
\cC(\sig,P)$ is $n^{O(1)}$. Similarly, $\Slong$ consists of $O(1)$ segments,
and the number of different subsets $\Slong$ that can arise over all sets
$S\in \cC(\sig,P)$ is $n^{O(1)}$.
}{
A segment in~$\Sshort$ has both endpoints at distance at most $\lsmall$
from~$\sig$, and so both endpoints are in $P_i(\sig)$ for $i=-\log_{3/2} n^{1-1/d}$.
By Corollary~\ref{coro:sep-point-distance}, the number of points in
this $P_i(\sig)$ is $O((3/2)^i n^{1-1/d}) = O(1)$. Hence, $|\Sshort|=O(1)$,
which trivially implies we can choose $\Sshort$ in $n^{O(1)}$ ways.
The number of segments in $\Slong$ is $O(1)$ by Packing Property~(PP1),
which again implies that we can choose $\Slong$ in $n^{O(1)}$ ways.
}

\noindent We can also enumerate the candidates for $\Sshort$ and $\Slong$ in polynomial time. It remains to handle~$\Smid$.

\myclaim{
For any $S\in \cC(\sig,P)$ the set $\Smid$ consists of $O(n^{1-1/d})$ segments,
and the number of different subsets $\Smid$ that can arise over all sets $S\in \cC(\sig,P)$ is $2^{O(n^{1-1/d})}$.
}{
Define $\Smid(i)\subseteq \Smid$ to be the following:

\[\Smid(i) \eqdef \left\{ s\in\Smid \mid 2^{i-1}/n^{1/d} < \len(s) \leq 2^{i}/n^{1/d} \right\}.\]
Note that
\[\Smid = \bigcup \{ \Smid(i) \mid -\log_{3/2} n^{1-1/d}+1 \leq i\leq \log n^{1/d} \}.\]
We first analyze $|\Smid(i)|$ and the number of ways in which we
can choose $\Smid(i)$ for a fixed~$i$.
To this end, we partition each face~$f$ of $\sig$ into a $(d-1)$-dimensional
grid whose cells have \replaced{edge length}{size} $2^i/n^{1/d}$. (If $n^{1/d}/2^i$ is not an integer,
then we use \replaced{cells of edge length}{size} $1/\lceil n^{1/d}/2^i \rceil$; all subsequent arguments work in this case as well.)
Let $G_i$ be the set of
grid points generated over all faces~$f$, and note that
\[|G_i|=O((n^{1/d} / 2^i)^{d-1})=O(n^{1-1/d} / 2^{i(d-1)}).\]
For each $g\in G_i$,
let $H_g$ denote the axis-aligned hypercube of size $2^{i+1}/n^{1/d}$
centered at~$g$; see Fig.~\ref{fig:grid}. Let
$\cH_i \eqdef \{ H_g \mid g\in G_i \}$ be the set of all these hypercubes.
Note that for any segment $s\in\Smid(i)$ there is a hypercube
$H_g\in \cH_i$ that contains~$s$.
Furthermore, all points in
any~$H_g$ have distance at most $2^i/n^{1/d}$ from $\sig$,
and so $P\cap H_g \subseteq P_i(\sig)$ for all~$g$.
\begin{figure}
\begin{center}
\includegraphics{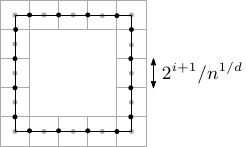}
\end{center}
\caption{The grid~$G_i$ used in the proof of Theorem~\protect\ref{thm:sep-tsp}.
               The grid points are shown in black and gray;
               only the hypercubes $H_g$ of the gray grid points are shown.}
\label{fig:grid}
\end{figure}

Now let $n_g$ denote the number of points from $P$ inside~$H_g$.
Since each point $p\in P$ is contained in a constant (depending on $d$)
number of hypercubes~$H_g$, we have $\sum_{g\in G_i} n_g = O(|P_i(\sig)|)$.
Furthermore, by Packing Property~(PP2) we know that a hypercube~$H_g$
can contain only $O(1)$ segments from  $\Smid(i)$. Thus
\begin{align*}
|\Smid(i)| &= O(\mbox{number of non-empty hypercubes $H_g$}) \\
           &= O(\min(|G_i|,|P_i(\sig)|))\\
           &= O(\min(n^{1-1/d} / 2^{i(d-1)}, |P_i(\sig)|)).
\end{align*}
For $i<0$ we have $|P_i(\sig)| = O((3/2)^i n^{1-1/d})$, which implies
\begin{align*}
|\Smid|
&= \sum_{i} |\Smid(i)| \\
&= O\left( \sum_{i<0} (3/2)^i n^{1-1/d} + \sum_{i\geq 0} n^{1-1/d}/2^{i(d-1)} \right)
= O(n^{1-1/d}).
\end{align*}
Now consider the number of ways in which we can choose $\Smid(i)$.
In a hypercube~$H_g$ \replaced{$\Smid(i)$ has}{we have} $O(1)$ edges which we can choose in
$n_g^{O(1)}$ ways. Hence, if $G^*_i\subseteq G_i$ denotes the collection of grid points~$g$ such
that $n_g>0$ (in other words, such that $H_g$ is non-empty), then
\[
\mbox{ total number of ways to choose~$\Smid(i)$ }
 = \prod\limits_{g\in G^*_i} n_g^{O(1)}
 = 2^{O(\sum_{g\in G^*_i} \log n_g)}.
\]
We bound $\sum_{g\in G^*_i} \log n_g$ separately for $i\geq 0$ and $i<0$.

First consider the case $i\geq 0$. Here we have
\[\sum_{g\in G^*_i} n_g = O(|P_i(\sig)|)=O(4^i n^{1-1/d}).\]
Moreover, $|G^*_i|\leq |G_i| = O\left(\frac{n^{1-1/d}}{2^{i(d-1)}}\right)$ and
$\frac{n^{1-1/d}}{2^{i(d-1)}} \leq 4^i n^{1-1/d}$. Since $G^*_i$ \replaced{contains the
grid points corresponding to}{are} the
non-empty cells, we have $|G_i^*| < c|P_i(\sig)|$ for all $i$, where $c$
is the maximum number of cubes in $\cH_i$ with a non-empty common intersection.
Note that $c$ is a constant depending only on $d$.
Hence,
\begin{align*}
\sum_{g\in G^*_i} \log n_g
& \leq |G^*_i| \cdot \log \left(\frac{|P_i(\sig)|}{|G^*_i|}\right)\\
&< |G^*_i| \cdot \log \left(\frac{c\cdot e\cdot |P_i(\sig)|}{|G^*_i|}\right) \\
& = O\left( \frac{n^{1-1/d}}{2^{i(d-1)}} \cdot \log 2^{i(d+1)} \right)\\
&= O\left( \frac{i(d+1)}{2^{i(d-1)}} \cdot n^{1-1/d} \right).
\end{align*}
In this formula, the first step follows from the AM-GM inequality. The third
step uses that $x\log(c\cdot e\cdot |P_i(\sig)|/x)$ is monotone increasing
for $x\in \oiv{0,c|P_i(\sig)|}$, therefore we can replace
$|G^*_i|$ with $|G_i|$ (since $|G^*_i| < |G_i|$).

Now consider the case~$i<0$. The number of points in the hypercubes is
\[\sum_{g\in G^*_i} n_g = O(|P_i(\sig)|)=O((3/2)^i n^{1-1/d}).\]
Because for $i<0$ we have $n^{1-1/d}/2^{i(d-1)} > (3/2)^i n^{1-1/d}$, the number
of points to distribute is smaller than the number of available
hypercubes, and so $\sum_{g\in G^*_i} \log n_g$ is maximized when
$n_g=2$ for all $g$ (except for at most one grid point~$g$).
Hence,
\[
\sum_{g\in G^*_i} \log n_g = O(|P_i(\sig)|) = O((3/2)^i n^{1-1/d}).
\]
Thus the total number of ways in which we can choose $\Smid$ is bounded by
the following expression, where $i$ ranges from $\imin \eqdef -\log_{3/2} n^{1-1/d} +1$ to $\imax \eqdef \log n^{1/d}$:
\begin{align*}
 \hspace{1cm}\prod_{i} &\expO{\sum_{g\in G^*_i} \log n_g} \\
  & =
\prod_{i<0} \expO{(3/2)^i n^{1-1/d}}
 \prod_{i\geq 0} \expO{\frac{i(d+1)}{2^{i(d-1)}} \cdot n^{1-1/d}}& \\
  & =
 \expO{ \sum_{i<0}  (3/2)^i n^{1-1/d}  +
 \sum_{i\geq 0}  \frac{i(d+1)}{2^{i(d-1)}} \cdot n^{1-1/d} }& \\
& = 2^{O(n^{1-1/d})}. &
\end{align*}
\added{This concludes the proof of the claim.}
}

In order to enumerate candidates for $\Smid$, we can combine all segment choices from the hypercubes $\cH_i$ for each $i$ that were considered in the above claim. There are $2^{O(n^{1-1/d})}$ possible combinations, and they can be enumerated
replaced{in polynomial time}{with polynomial delay}.

Combining these sets with the candidates for $\Sshort$ and $\Slong$, which we already showed how to enumerate earlier,  yields the desired collection $\cC'(\sig,P)$. The size of any combination of three candidate subsets $\Sshort,\Smid,\Slong$ is $O(1) + O(n^{1-1/d}) + O(1)$, and the total number of combinations is $n^{O(1)} \cdot 2^{O(n^{1-1/d})} \cdot n^{O(1)} = 2^{O(n^{1-1/d})}$. The time to enumerate all combinations is asymptotically the same.
\end{proof}


\section{An exact algorithm for TSP}\label{sec:alg}
In this section we design an exact algorithm for TSP using the separator theorem from the previous section.
\replaced{This theorem}{As a first step, let us take a look at the TSP problem in $\Reals^2$.
The separator theorem from the previous section} provides us with a
separator~$\sig$ such that the number of segments from an optimal tour
that cross~$\sig$ is $O(n^{1-1/d})$. Moreover, the number of candidate
subsets $S\in \cC(\sig,P)$ that we need to try is only $2^{O(n^{1-1/d})}$.
We can now obtain a divide-and-conquer algorithm similar to the algorithms
of~\cite{Kann92,SmithW98} in a relatively standard manner. As we shall see, however,
the resulting algorithm would still not run in $2^{O(n^{1-1/d})}$ time.
We will therefore need to modify the algorithm and employ the so-called rank-based approach~\cite{single-exponential} to get our final result. In what follows, we describe an exact algorithm for TSP in $\Reals^d$.
\smallskip

A separator-based divide-and-conquer algorithm for \etsp works
as follows. We first compute a separator using Theorem~\ref{thm:sep-tsp}
for the given point set. For each candidate subset of edges crossing the separator,
we then need to solve a subproblem for the points inside the separator and
one for the points outside the separator.
In these subproblems we are no longer searching for a shortest tour,
but for a collection of paths \added{of minimum total length} that connect
\added{the endpoints of} the edges crossing the separator in a suitable manner.

To define the subproblems more precisely,
let $P$ be a point set and let $M$ be a perfect matching on a set~$B\subseteq P$ of
so-called \emph{boundary points}. We say that a
collection~$\cP=\{\pi_1,\ldots,\pi_{|B|/2}\}$ of paths \emph{realizes $M$ on $P$}
if
(i) for each pair $(p,q)\in M$ there is a path $\pi_i\in \cP$ with~$p$ and~$q$
as endpoints,
and (ii) the paths together visit each point $p\in P$ exactly once.
We define the length of a path $\pi_i$ to be the sum of the Euclidean lengths of its edges, and we define the total length of $\cP$ to be the sum of the lengths of the paths~$\pi_i\in\cP$.
The subproblems that arise in our divide-and-conquer algorithm can now be defined
as follows; \added{see Fig.~\ref{fig:path-cover} for an illustration.}
\begin{figure}[t]
\centering
\includegraphics{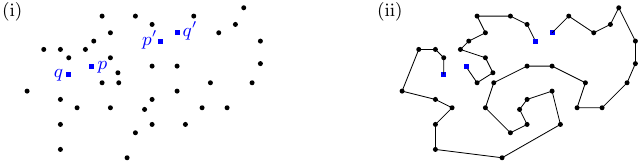}
\caption{(i) An instance of \btsp with $B=\{p,q,p',q'\}$ and $M =\{(p,p'),(q,q')\}$.
         (ii) Solution for the instance from (i). }
\label{fig:path-cover}
\end{figure}

\defproblem{\btsp}
{A point set $P\subset \Reals^d$, a set of boundary points $B\subseteq P$, and a perfect matching $M$ on $B$.}
{Find a collection of paths covering $P$ of minimum total length that realizes~$M$ on~$P$.}
\added{Note that the matching in $M$ does \emph{not} specify edges that must be in the final tour,
but it specifies which pairs of points form the endpoints of the paths in the cover.
Also} note that we can solve \etsp on a point set $P$ by creating
a copy $p'$ of an arbitrary point~$p\in P$, and then solve \btsp on $P\cup\{p'\}$ with
$B\eqdef\{p,p'\}$ and $M \eqdef \{(p,p')\}$.

\added{We will solve this initial instance of  \btsp recursively, as explained in detail below.
Roughly speaking, we will compute a separator $\sigma$ together with its candidate
sets $\cC’(\sig,P)$ of edges crossing the separator, solve suitable recursive problems
inside and outside the separator, and then glue the solutions to the subproblems together to
obtain a solution to the overall problem. (Our actual algorithm will be more involved and use the rank-based approach, to deal with the matching in an efficient manner.)  Note that the ``correct’’ candidate set $S\in\cC’(\sig,P)$
leads to subproblems whose solutions are part of an optimal TSP tour and, hence, satisfy the Packing Property.
``Wrong’’ candidate sets, however,  lead to other subproblems and it is not clear if an optimal
solution to these subproblems satisfies the Packing Property. Hence, our algorithm may not find
an optimal solution to such an instance of \btsp. It might even happen that no valid solution
satisfying the Packing Property exists for certain subproblems. Fortunately this is not a problem,
because of the following two properties: (i) we know that there will be a sequence of recursive calls
that are consistent with an optimal solution for the original TSP problem, and this sequence will
be solved correctly, (ii) the solution (if any)  that is reported for any of the subproblems
is a valid solution. Thus we are guaranteed that we find an optimal TSP tour, and that no
shorter but invalid tour can be reported. With this issue out of the way, we now describe our algorithm in detail.}
\medskip

An instance of \btsp is solved by a separator-based recursive algorithm
as follows. Let $(P,B,M)$ be an instance of \btsp, and let $\sig$ be a
separator for $P$. We consider each candidate set $S\in\cC'(\sig,P)$  of edges
crossing the separator~$\sig$. In fact, it is sufficient to consider candidate
sets where the number of segments from $S$ incident to any point $p\in
P\setminus B$ is at most two, and the number of segments from $S$ incident to
any point in $B$ is at most one.

We now wish to define subproblems for $\sigin$ and $\sigout$ (the regions
inside and outside $\sig$, respectively) whose  combination yields a solution
for the given problem on~$P$. Let $P_1(S)\subseteq P$ denote the set of
endpoints with precisely one incident segment from~$S$, and let
$P_2(S)\subseteq P$ be the set of endpoints with precisely two incident
segments from~$S$. Note that in a solution to the problem $(P,B,M)$ the points
in $B$ need one incident edge---they must become endpoints of a path---while
points in $P\setminus B$ need two incident edges. This means that the points
in $B\cap P_1(S)$ and the points in $P_2(S)$ now have the desired number of
incident edges, so they can be ignored in the subproblems. Points in  $B\sd
P_1(S) \eqdef (B\setminus P_1(S)) \cup (P_1(S) \setminus B)$ still need one
incident edge, while the remaining points in $P\setminus \left((B\cap P_1(S)) \cup
P_2(S)\right)$ still need two incident edges. Hence, for $\sigin$ we obtain
subproblems of the form $(\Pin,\Bin,\Min)$  where
\begin{equation}\label{eq:pindef}
\begin{split}
\Pin &\eqdef \Big(P \setminus \big((B\cap P_1(S)) \cup P_2(S)\big)\Big)\cap \sigin, \\
\Bin &\eqdef \big( B\sd P_1(S) \big) \cap \sigin, \\
\Min &\text{ is a perfect  matching on }\Bin.
\end{split}
\end{equation}
\added{See Figure~\ref{fig:btsp}.}
\begin{figure}[t]
\centering
\includegraphics{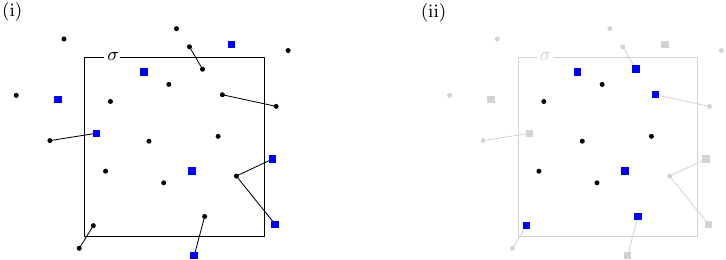}
\caption{(i) A \btsp instance and a separator $\sig$. Blue squares indicate boundary points. (ii) The subproblem generated inside~$\sig$ consists of the black points and blue squares.)}\label{fig:btsp}
\end{figure}
For $\sigout$ we obtain subproblems of the form $(\Pout,\Bout,\Mout)$
defined in a similar way. As already remarked, we can restrict our attention to
candidate sets $S\in \cC(\sig,P)$
that contain at most one edge incident to any given point in~$B$,
and at most two edges incident to any given point in~$P\setminus B$.
Moreover, $S$ should be such that $|\Bin|$ and $|\Bout|$ are even.
We define $\cC^*(\sig,P)$ to be the family of candidate sets $\cC'(\sig,P)$ restricted this way.
Also note that while $\Pin,\Bin$ and $\Pout,\Bout$ are determined
once $S$ is fixed, the algorithm has to find the best matchings
$\Min$ and $\Mout$. These matchings should together realize the
matching $M$ on $P$ and, among all such matchings, we want the pair
that leads to a minimum-length solution.

The number of perfect matchings on $k$ points is $k^{\Theta(k)}$.
Unfortunately, already in the first call of the recursive algorithm $|\Bin|$
and $|\Bout|$ can be as large as $\Theta(n^{1-1/d})$.  Hence, recursively
checking all matchings will not lead to an algorithm with the desired running
time. In $\Reals^2$ we can use that an optimal TSP tour is crossing-free,  so
it is sufficient to look for ``crossing-free matchings'', of which there are
only $2^{O(k)}$. (This approach would actually require a different setup of
the subproblems; see the papers by Deineko~\etal~\cite{DeinekoKW06} and
Dorn~\etal~\cite{DornPBF10}.) However, the crossing-free property has no
analogue in higher dimensions, and it does not hold in $\Reals^2$ for our
``almost-Euclidean'' setting either. Hence, we need a different approach to
rule out a significant proportion of the available matchings.

\paragraph{Applying the rank-based approach}
Next we describe how we can use the rank-based
approach~\cite{single-exponential,CyganKN18} in our setting. We try to avoid the
intricate notation introduced in the original papers, but our terminology is
mostly compatible with~\cite{single-exponential}. A standard
application of the rank-based approach works on a tree-decomposition  of the
underlying graph, where the bags represent vertex separators.
In our application the underlying graph is a complete graph on the
points---all segments  are potentially segments of the TSP tour---and   we
have to use a separator for the edges in the solution.

Let $P$ be a set of points in $\Reals^d$, and let $B\subseteq P$ be a set of
boundary points such that $|B|$ is even. Let $\cM(B)$ denote the set of all perfect
matchings on~$B$, and consider a matching $M\in \cM(B)$. We can turn $M$ into
a weighted matching by assigning to it the minimum total length of any
solution realizing $M$. In other words, $\weight(M)$ is the length of the
solution of \btsp for  input~$(P,B,M)$. Whenever we speak of weighted
matchings in the sequel, we always mean perfect matchings on a set $B\subseteq
P$ weighted as above. We
use $\cM(B,P)$ to denote the set of all such weighted matchings on~$B$. Note
that $|\cM(B,P)| = |\cM(B)| = 2^{O(|B|\log |B|)}$. The key to reducing the
number of matchings we have to consider is the concept of representative sets,
as explained next.

We say that two matchings $M,M' \in \cM(B)$ \emph{fit} if their union is a
Hamiltonian cycle \added{on~$B$}. Consider a pair $P,B$. Let $\cR$ be a set of  weighted
matchings on~$B$ and let $M$ be another matching on~$B$. We define
\[\opt(M,\cR) \eqdef \min \{ \weight(M') \mid M' \in \cR\text{ and } M' \mbox{ fits }
M\},\]
that is, $\opt(M,\cR)$ is the minimum total length of any collection of
paths on $P$ that together with the matching $M$ forms a cycle.

A set $\cR \subseteq \cM(B,P)$ of weighted matchings  is defined to be\index{representative set}
\emph{representative} of another set $\cR' \subseteq \cM(B,P)$ if for any
matching $M\in \cM(B)$ we have $\opt(M,\cR) = \opt(M,\cR')$. Note that our
algorithm is not able to compute a representative set of $\cM(B,P)$, because
it is also restricted by the Packing Property, while a solution of \btsp
for a generic $P,B,M$ may not satisfy it. Let $\cM_\PP(B,P)$ denote
the set of weighted matchings in $\cM(B,P)$ that have a corresponding \btsp
solution satisfying the Packing Property.

The basis of the rank-based method is the following result.
\begin{lemma} {\rm\bf [Bodlaender~\etal~\cite[Theorem 3.7]{single-exponential}]} \label{lem:repr}
There exists a set $\cR^*$ consisting of $2^{|B|-1}$ weighted matchings that is representative of the set $\cM(B,P)$.  Moreover, there is an algorithm \emph{Reduce} that, given
a representative set~$\cR$ of $\cM(B,P)$, computes such a set $\cR^*$ in
$|\cR| \cdot 2^{O(|B|)}$ time.
\end{lemma}
Lemma~\ref{lem:repr} also holds for our case, where $\cR$ is representative of $\cM_\PP(B,P)$. The result of Bodlaender~\etal{} is actually more general than stated above,
as it not only applies to matchings but also to other types of partitions.
Moreover, for matchings the bound has been improved to $2^{|B|/2-1}$~\cite{CyganKN18}.
However, Lemma~\ref{lem:repr} suffices for our purposes.

\added{We now briefly sketch the main ideas behind the rank-based approach. Consider a matrix $\cA$ whose rows and columns are indexed by perfect matchings of $B$. One can think of the row index as the matching realized by the subproblem inside the separator, and the column index as the matching realized by the subproblem outside the separator. An entry in $\cA$ is $1$ if the pair of corresponding matchings fit, that is, their union is a Hamiltonian cycle on $B$, and the entry is $0$ otherwise.

Consider a set $X$ of rows in $\cA$ that is linearly dependent over the field $\mathrm{GF}(2)$. Suppose that $M$ is a matching in $X$ of largest weight.
Bodlaender et al.~\cite{single-exponential} show that $X\setminus \{M\}$ is a representative set for the set $X$ --- this is because every matching (column) that combined with $M$ gives a solution that can also be combined with a matching (row) in $X\setminus \{M\}$ to a solution (of equal or lower weight). Thus, when we have a
set of rows in $\mathcal{A}$ that is linearly dependent over the field $\mathrm{GF}(2)$, then the heaviest matching in this set is redundant and can be removed.
Consequently, it is sufficient to compute a minimum weight row basis of $\cA$, as the matchings corresponding to the rows in such a basis form a representative set. The basis can be computed using standard linear algebra; this is done by the Reduce algorithm of Lemma~\ref{lem:repr}. In order to bound the size of the obtained representative set, it is enough to bound the rank of $\cA$. Bodlaender~\etal{} prove that $\cA$ has rank at most $2^{|B|-1}$ by providing a factorization $\cA =  \hat\cA^T\hat \cA$ where $\hat \cA$ has $2^{|B|-1}$ rows.}

Lemma~\ref{lem:repr} bounds the size of the representative set in terms
of~$|B|$, the number of boundary points. In the first call of our algorithm
$|B| = O(n^{1-1/d})$ because of the properties of our separator, but we have
to be careful that the size of~$B$ stays \replaced{sufficiently small}{under control} in recursive calls.

Algorithm~\ref{alg:TSPdd} on page~\pageref{alg:TSPdd} describes how we deal
with this. A key step in the algorithm is Step~\ref{step:balance}, where we invoke the
balance condition of the separator with respect to $B$ or $P$ depending on the
size of $B$ relative to the size of $P$. (The constant $\gamma$ will be
specified in the analysis of the running time.) Steps~\ref{step:combine}--Steps~\ref{step:reduce} combine
the representative sets  $\cR_\myin$ and $\cR_\myout$. Next we explain these steps in more detail.
\begin{algorithm}[t]
\begin{algorithmic}[1]
\Statex \textbf{Input:} A set $P$ of points in $\Reals^d$ and a subset $B\subseteq P$
\Statex \textbf{Output:} A set $\cR\subset \cM(B,P)$ with $|\cR| \leq 2^{|B|-1}$ representing $\cM_\PP(B,P)$
\If{$|P| \leq 1$}
   \Return $\{(\emptyset,0)\}$
\EndIf
  \State $\cR \gets \emptyset$
  \State \label{step:balance}Compute~$\sig$ (by Thm.~\ref{thm:sep-tsp}) with $Q\!=\!P$ if
         $|B| \!\leq\! \gamma|P|^{1-1/d}$ and $Q\!=\!B$ otherwise.
         \ForAll{candidate sets $S\in \cC^*(\sig,P)$ } \label{step:outerloop}
         	\State $\cR_\myin \gets $ \emph{TSP-Repr}($\Pin,\Bin$), where $\Pin$ and $\Bin$ are defined as in~\eqref{eq:pindef}
         	\State $\cR_\myout \! \gets\!$ \emph{TSP-Repr}($\Pout,\Bout$), where $\Pout$ and $\Bout$ are defined as in~\eqref{eq:pindef}
         	\ForAll{$\Min \in \cR_\myin$ and $\Mout \in \cR_\myout$} \label{step:combine}
            	\If{$\Min$ and $\Mout$ are compatible} 	  		
                \State \label{step:insert} Add $(\mathrm{Join}_S(\Min,\Mout), \weight(\Min) \!+\! \len(S) \!+\! \weight(\Mout))$ to $\cR$
                 \EndIf
         	\EndFor
           \EndFor
  \State \label{step:reduce} $\cR \gets$ \emph{Reduce}$(\cR)$
\State \Return{$\cR$}
\end{algorithmic}
\caption{\emph{TSP-Repr}$(P,B)$}\label{alg:TSPdd}
\end{algorithm}

Consider a set $S\in \cC^*(\sig,P)$, a matching $\Min \in \cM(\Bin)$ and a
matching $\Mout \in \cM(\Bout)$. Let $G=G_S (\Min, \Mout)$ be the graph with
vertex set $V(G) \eqdef B\cup P_1(S) \cup P_2(S)$ and edge set $E(G) \eqdef
\Min \cup \Mout\cup S$.  We say that $\Min$ and $\Mout$ are \emph{compatible}\index{compatible matchings}
if $G$ consists of $|B|/2$ disjoint paths covering $V(G)$ whose endpoints are
exactly the points in~$B$. A pair of compatible matchings  induces a perfect
matching on $B$, where for each of these $|B|/2$ paths  we add a matching edge
between its endpoints. We denote this matching by $\mathrm{Join}_S(\Min,
\Mout)\in \cM(B)$. To get a set~$\cR$ of weighted matchings on $B$ we thus
iterate in Steps~\ref{step:combine}--\ref{step:insert} through all pairs
$\Min,\Mout$ where $\Min$ and $\Mout$ are compatible, and for such pairs, we
add to $\cR$ the matching $\mathrm{Join}_S(\Min,\Mout)$. The weight of this
matching is $\weight(\Min)+\len(S)+\weight(\Mout)$.

\added{Recall that the initial call to \emph{TSP-Repr}$(P,B)$ is done by creating
a copy $p'$ of an arbitrary point~$p\in P$, and then solving \btsp on $P\cup\{p'\}$ with $B\eqdef\{p,p'\}$ and $M \eqdef \{(p,p')\}$. This solves \etsp as desired, since the representative set $\cR$ consists of a single matching $\{(p,p')\}$ whose weight is the optimum Euclidean TSP tour of $P$. One can use standard techniques to obtain the optimum tour itself with a slightly modified algorithm. The correctness of our algorithm is implied by the following claim.}

\myclaim{
The set $\cR$ created in Lines~\ref{step:outerloop}--\ref{step:insert} of Algorithm~\ref{alg:TSPdd} is representative of $\cM_\PP(B,P)$.
}{
The proof is by induction on $|P|$. Clearly, for $|P|\leq 1$ the claim holds. Otherwise, let $S\in \cC^*(\sigma,P)$ be fixed. The set $S$ is considered in some iteration of the outer loop. Define $\Pin,\Pout,\Bin,\Bout$ as in this iteration, and let $\cR_\myin,\cR_\myout$ be the sets returned by the recursive calls, which are representative sets  of $\cM_\PP(\Bin,\Pin)$ and $\cM_\PP(\Bout,\Pout)$ respectively by induction. Notice that $S$ can be regarded as a \btsp solution for $(\Bin\cup\Bout,\Bin\cup\Bout,M_S)$, where $M_S$ is the matching realized by $S$ on $\Bin\cup\Bout$. Let $\len(S)$ be the weight assigned to $M_S$. Clearly $\{M_S\}$ is representative of $\{M_S\}$. Now our $\mathrm{Join}_S$ operation can be regarded as the succession of two join operations as defined by~\cite{single-exponential}, applied to $\{M_S\}$ and $\cR_\myin$ first, and then to the result and $\cR_\myout$ second. By Lemma 3.6 in~\cite{single-exponential}, the join operation preserves representation, therefore the matchings added to $\cR$ in this iteration of the outer loop form a representative set of
\begin{align*}
\widehat{\cM}_S \eqdef \{\mathrm{Join}_S(\Min,\Mout) \;|\; &\Min \in \cM_\PP(\Bin,\Pin),\\ &\Mout \in \cM_\PP(\Bout,\Pout)\big\}.
\end{align*}
Consequently, the set $\cR$ that is created at the end of the outer loop is a representative set of $\widehat{\cM} \eqdef \bigcup_{S \in \cC^*(P,\sig)} \widehat{\cM}_S$.
The set $\widehat{\cM}_S$ contains the subset of $\cM_\PP(B,P)$ that has a corresponding optimum  with the Packing Property that intersects $\sigma$ in $S$, because for any such optimum path cover $\cP$, the subpaths of $\cP$ induced by $\Pin$ also have the Packing Property, and form an optimal  \btsp for the input $(\Pin, \Bin, \Min)$, i.e., there is a corresponding weighted matching $\Min\in \cM_\PP(\Bin,\Pin)$. (The analogous statement is true for the subpaths induced by $\Pout$.) Since $\cC^*(\sigma,P)$ contains all sets $S$ that can arise as the set of segments intersecting $\sig$ in an optimum \btsp solution with the Packing Property, it follows that $\widehat{\cM} \supseteq \cM_\PP(B,P)$, which concludes the proof of the claim.
}

Notice that lines \ref{step:outerloop}-\ref{step:insert} can be implemented using a brute-force algorithm; this takes
\[O\left(|\cR_\myin| \cdot |\cR_\myout| \cdot \poly(|B|+|S|)\right)\] time.
However, by combining $\cR_\myin$ and $\cR_\myout$ in this way, the size of $\cR$
may be more than~$2^{|B|-1}$.
Hence, we apply the \emph{Reduce} algorithm~\cite{single-exponential},
to create a representative set of size at most $2^{|B|-1}$
in $|\cR|\cdot 2^{O(|B|)}$ time. Since our recursive algorithm
ensures that $|\cR_\myin| \leq 2^{|\Bin|-1}$ and $|\cR_\myout|\leq 2^{|\Bout|-1}$,
all of the above steps run in $2^{O(|B|+|S|)}=\expO{|B|+|P|^{1-1/d}}$ time.

\subsection{Detailed analysis of the running time}\label{subsec:analysis}
In this section, we provide a detailed analysis of the running time of our algorithm.

The running time of \emph{TSP-Repr}$(P,B)$ essentially satisfies the following recurrence,  where $c_0,c_1,c_2$ are positive constants and we use the notation $n\eqdef|P|$ and $b\eqdef|B|$.
\[
T(n,b) \leq
\begin{cases}
c_0 &\text{ if } n\leq 1\\
2^{c_1(n^{1-1/d}+b)}T\left(\delta n,\, b + c_2 n^{1-1/d}\right) &\text{ if } b \leq \gamma n^{1-1/d}\\
2^{c_1(n^{1-1/d}+b)}T\left(n,\, \delta b + c_2 n^{1-1/d}\right) &\text{ if } b > \gamma n^{1-1/d},
\end{cases}
\]
The actual recurrence is a bit more subtle, as explained next.

For each $S\in \cC^*(\sig,P)$, let $n_{S,\myin}\eqdef |\Pin|$,
let $b_{S,\myin} \eqdef |\Bin|$, let $n_{S,\myout} \eqdef|\Pout|$,
and let $b_{S,\myout}\eqdef |\Bout|$. By the discussion in Section~\ref{sec:alg}, we can bound the running time of the two inner loops, the \emph{Reduce} algorithm and the rest of the operations outside the recursive calls by $\exp(c_3(n^{1-1/d}+b))$ for some positive constant $c_3$. Therefore, the algorithm \emph{TSP-Repr}$(P,B)$ obeys the following recursion,
where $\sig_P$ and $\sig_B$ are separators balanced with respect to $P$ and $B$,
respectively.
\begin{multline*}
T(n,b) \leq \\
\begin{dcases}
\hspace{0.5cm}c_0 & \text{ if } n\leq 1 \\[.6em]
\!\sum_{S \in \cC^*(\sig_P,P)} \hspace{-0.4cm}\big( \exp(c_3(n^{1-1/d}+b)) \!+\! T(n_{S,\myin},b_{S,\myin})\!+\!T(n_{S,\myout},b_{S,\myout})\big) & \text{ if } b \leq \gamma n^{1-1/d}\\
\!\sum_{S \in \cC^*(\sig_B,P)} \hspace{-0.4cm} \big(\exp(c_3(n^{1-1/d}+b)) \!+\! T(n_{S,\myin},b_{S,\myin})\!+\!T(n_{S,\myout},b_{S,\myout})\big) & \text{ if } b > \gamma n^{1-1/d}.
\end{dcases}
\end{multline*}
Note that the terms in the second and third case are the same, except that the second case uses a separator $\sig_P$ while the third case uses a separator $\sig_B$. This in turn will influence the bounds on $n_{S,\myin}$ and $b_{S,\myin}$ (and, similarly, $n_{S,\myout}$ and $b_{S,\myout}$).
\begin{lemma}
\added{For a suitable choice of $\gamma$ in Step~\ref{step:balance} of \emph{TSP-Repr} we have}
$T(n,2) = 2^{O(n^{1-1/d})}$.
\end{lemma}
\begin{proof}
We prove by induction that $T(n,b) \leq \exp(d_1 n^{1-1/d} + d_2 b)$ for some constants $d_1$ and $d_2$ and for all $1\leq b \leq n$. This clearly holds for $b,n\leq 1$, \added{when, as we will ensure, we have $d_1,d_2\geq c_0$}. Hence, by induction, for each $S$ we have

\begin{align*}
&  \exp(c_3(n^{1-1/d}+b)) + T(n_{S,\myin},b_{S,\myin})+T(n_{S,\myout},b_{S,\myout})\\
& \hspace{0.5cm}  \added{\leq \exp\left(c_3(n^{1-1/d}+b)\right)+\exp\left(d_1 n_{S,\myin}^{1-1/d} + d_2 b_{S,\myin}\right) + \exp\left(d_1 n_{S,\myout}^{1-1/d} + d_2 b_{S,\myout}\right)}\\
& \hspace{0.5cm} \added{\leq \exp\left(c_3(n^{1-1/d}+b)\right)+2\exp\left(d_1 n_{S,\max}^{1-1/d} + d_2 b_{S,\max}\right),}
\end{align*} 
\replaced{where $n_{S,\max} = \max(n_{S,\myin},n_{S,\myout})$ and $b_{S,\max} = \max(b_{S,\myin},b_{S,\myout})$.}{where the additions were replaced by multiplications.} 
\added{Note that for $x,y\geq 2$ we have $\exp(x)+2 \exp(y) \leq \exp(x+y)$. Therefore, since $c_3>2$ and we may assume that $d_1,d_2>2$, we can conclude that
\begin{align*}
&\exp(c_3(n^{1-1/d}+b)) + T(n_{S,\myin},b_{S,\myin})+T(n_{S,\myout},b_{S,\myout})\\
& \hspace{2cm} \leq \exp\left(c_3(n^{1-1/d}+b) + d_1 n_{S,\max}^{1-1/d} + d_2 b_{S,\max}\right).
\end{align*}
}
Let $c_2$ and $c_4$ be the constants from part (ii) and (iii) of Theorem~\ref{thm:sep-tsp}. For any $S\in \cC^*(\sig_B,P)$, we have \replaced{$b_{S,\max}\leq \delta b  + c_2 n^{1-1/d}$}{ $b_{S,\myin}\leq \delta b  + c_2 n^{1-1/d}$ and $b_{S,\myout}\leq \delta b + c_2 n^{1-1/d}$}; similarly, for any $S\in \cC^*(\sig_P,P)$, we have \replaced{$n_{S,\max}\leq \delta n$}{ $n_{S,\myin}\leq \delta n$ and $n_{S,\myout}\leq \delta n$}. Note that we always have the trivial bounds $b_{S,\added{\max}}\leq b  + c_2 n^{1-1/d}$ and $n_{S,\added{\max}}\leq n$ as well. Since $|\cC^*(\sig,P)| \leq \exp(c_4 n^{1-1/d})$, we get the following:

\[
T(n,b) \leq
\begin{cases}
c_0 & \text{ if } n\leq 1\\[.6em]
\exp\big(c_1(n^{1-1/d}+b)  +d_1 (\delta n)^{1-1/d} + d_2(b + c_2 n^{1-1/d})\big)
& \text{ if } b \leq \gamma n^{1-1/d}\\[.6em]
\exp\big(c_1(n^{1-1/d}+b)  +d_1 n^{1-1/d} + d_2(\delta b + c_2 n^{1-1/d})\big)
& \text{ if } b > \gamma n^{1-1/d},
\end{cases}
\]
where $c_1 = c_3+c_4$. We set $c \eqdef \max(c_1,c_2)$, and let $\gamma \eqdef \frac{2c}{1-\delta}$. (Notice that the definition of $\gamma$ here is valid: it is independent of $d_1$ and $d_2$.)

If $b \leq \gamma n^{1-1/d} = \frac{2c}{1-\delta}n^{1-1/d}$, we have the following:
\begin{align*}
T(n,b) & \leq \exp\big(c n^{1-1/d} + cb + d_1(\delta n)^{1-1/d} + d_2b+d_2cn^{1-1/d})\big)\\
& \leq \exp\left(\left(c + \frac{2c}{1-\delta}+d_1\delta^{1-1/d}+d_2c\right)n^{1-1/d}+d_2 b\right)\\
&\leq \exp(d_1 n^{1-1/d} + d_2 b),
\end{align*}
where the second inequality uses $b \leq \frac{2c}{1-\delta}n^{1-1/d}$ and the third uses
\[c + \frac{2c}{1-\delta}+d_1\delta^{1-1/d}+d_2c \leq d_1.\]
This can be ensured by setting \[d_1 \eqdef \added{\max}\left(\added{c_0},\frac{c + 2c/(1-\delta) +d_2c}{1-\delta^{1-1/d}}\right).\]
Note that \replaced{}{$0<\delta<1$ and $c>0$ are fixed constants. Moreover, $d_2$ will be a positive constant as well; see below. Hence,} $d_1$ is a positive constant.

Finally, if $b > \gamma n^{1-1/d} = \frac{2c}{1-\delta}n^{1-1/d}$, we have the following:
\begin{align*}
T(n,b) & \leq \exp\big(c n^{1-1/d} + cb + d_1 n^{1-1/d} + d_2\delta b+d_2cn^{1-1/d})\big)\\
& < \exp\left(d_1 n^{1-1/d}+\left(\frac{1-\delta}{2} + c + \delta d_2 + \frac{1-\delta}{2}d_2\right) b\right)\\
&\leq \exp(d_1 n^{1-1/d} + d_2 b),
\end{align*}
where the strict inequality uses $cn^{1-1/d} <\frac{1-\delta}{2}b$, and the final inequality uses
\[\frac{1-\delta}{2}+c+ \frac{1+\delta}{2}d_2 \leq d_2.\]
We can ensure this by setting
\[d_2 \eqdef \added{\max}\left(\added{c_0},\frac{(1-\delta)/2+c}{1-(1+\delta)/2}\right),\]
therefore $d_2$ is a positive constant. Since there exists positive constants $d_1$ and $d_2$ satisfying the above inequalities, we have that
$T(n,b) \leq \exp(d_1 n^{1-1/d} + d_2 b)$, and in particular, for the initial call we have $T(n,2) = 2^{O(n^{1-1/d})}$.
\end{proof}


\section{Almost Euclidean TSP}\label{sec:packing}
So far we considered \etsp in the real-RAM model of
computation.
We now consider a slightly more general scenario in the Word-RAM model.
Here we assume that the input is a set $P$ of $n$ points in $\Reals^d$,
specified by rational coordinates, as well as a distance matrix~$D$. The basic
assumption we make is that the distances in $D$ approximate the real Euclidean distances
well. More precisely, we require that the ordering of pairwise distances on
the given point set $P \eqdef \{p_1,\ldots,p_n\}$ is
preserved: if $|p_i p_j| < |p_k p_\ell|$ then $D[i,j] < D[k,\ell]$. \added{We remark  that the entries of $D$ need not satisfy the triangle inequality.}
We call this the \emph{almost Euclidean} version of \TSP.

In order to show that our algorithms work in this setting, we only need to show that
an optimal tour in this setting satisfies Packing Property.
Note that the Packing Property for the almost Euclidean version immediately implies
that the Packing Property also holds for the Euclidean version (where, as remarked
in Section~\ref{sec:bal-sep}, similar properties were already known). Recall that
a set $S$ of segments in $\Reals^d$ has the \emph{Packing Property} if \index{Packing Property}
for any \replaced{hypercube}{separator}~$\sig$ we have
\begin{itemize}
\item (PP1):
  $\big| \{ s\in S \mid s \mbox{ crosses } \sig \mbox{ and } \len(s) \geq \size(\sig) \}  \big| \!=\! O(1)$ \\[-4mm]
\item (PP2):
  $\big|  \{ s\in S \mid s \subset \sigin \mbox{ and }  \len(s) \geq \size(\sig)/4 \}  \big| \!=\! O(1)$.
\end{itemize}

\begin{theorem}
Let $P\eqdef \{p_1,\ldots,p_n\}$ be a point set in $\Reals^d$ and let $D$ be a distance matrix for $P$
such that we have $|p_i p_j| < |p_k p_\ell|$ if and only if $D[i,j] < D[k,\ell]$. Let $T$ be a tour on $P$ that is
optimal for the distances given by~$D$. Then the set of edges of $T$ has the Packing Property.
\end{theorem}
\begin{proof}
We first prove Packing Property (PP1) and then argue that (PP2) follows from (PP1).

Let $\sig$ be a hypercube, and suppose without loss of generality that $\size(\sig) = 1$.
Suppose for a contradiction that there are more than $c$ tour edges of length at least~$1$ that cross~$\sig$,
where $c$ is a suitably large constant (which depends on $d$). By the pigeonhole principle,
we can then find three edges in $T$ such that
(i) the pairwise Euclidean distances between the endpoints of these edges that lie
    inside $\sig_\myin$ is at most $1/10$, and
(ii) the pairwise angles between these edges is at most $\pi/30$.
Here the angle between two edges is measured as the smaller angle between two lines going through
the origin and parallel to the given edges.

Now fix an orientation on the tour $T$ such that at least two of the three edges cross $\sig$
from inside to outside, and orient these edges accordingly. Let $p_i p_j$ and $p_k p_{\ell}$
denote these two oriented edges; see Figure~\ref{fig:pp1}.
\begin{figure}[t]
\centering
\includegraphics[width = 0.7\textwidth]{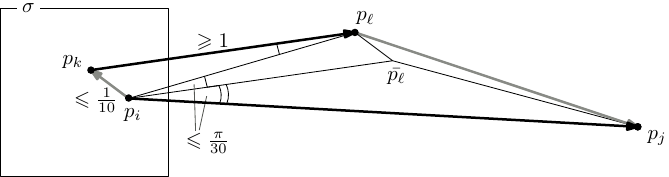}
\caption{Since $|p_i p_j|>|p_j p_{\ell}|$ and $|p_k p_{\ell}|>|p_i p_k|$, we can exchange $p_i p_j$
and $p_k p_{\ell}$ for $p_i p_k$ and $p_j p_{\ell}$ in the tour, and get a shorter tour.
Note for $d>2$ the points do not all have to lie in the same plane.}\label{fig:pp1}
\end{figure}
Thus $p_i,p_k\in \sig_\myin$ and $|p_i p_k|\leq 1/10$.
Assume without loss of generality that $|p_i p_j| \geq |p_k p_{\ell}|$.

Our goal now is to show that $|p_k p_{\ell}|>|p_i p_k|$ and $|p_i p_j|>|p_j p_{\ell}|$;
this will help us find a way to reroute the current tour into a shorter one, leading to a contradiction.
The first inequality is easy: $|p_k p_{\ell}|>|p_i p_k|$ since $|p_k p_{\ell}|\geq 1$ and $|p_i p_k|\leq 1/10$.
In order to show $|p_i p_j|>|p_j p_{\ell}|$, we proceed as follows.

In the triangle $p_i p_k p_{\ell}$, we have $|p_k p_{\ell}|\geq 1$ and $|p_i p_k|\leq 1/10$,
so by the law of sines,
$\sphericalangle (p_i p_{\ell} p_k) \leq \arcsin (1/10) < \pi/30$.
Let $\bar{p}_{\ell}=p_i+(p_{\ell}-p_k)$. Here and in the sequel we use that a point can also be thought of as a vector,
so we can add and subtract points to get new points.
Then we have that $\sphericalangle (p_{\ell} p_i \bar{p}_{\ell}) = \sphericalangle (p_i p_{\ell}p_k) \leq \pi/30$.
Due to our choice of $p_i p_j$ and $p_k p_{\ell}$ their angle is at most $\pi/30$, so we have
\[
\sphericalangle (\bar{p}_{\ell}p_i p_j ) = \sphericalangle \big((p_j-p_i),(p_{\ell}-p_k)\big)  \leq \pi/30.
\]
Therefore, $\sphericalangle (p_{\ell} p_i p_j) \leq \pi/15$. Note that this is also true if the points do
not all lie in the same plane.
Now observe that
\[
 |p_i p_{\ell}| \leq |p_k p_{\ell}| + |p_i p_k| \leq |p_i p_j|+1/10,
\]
that is, $p_i p_{\ell}$ cannot be much longer than $p_i p_j$.
Thus, if we look at the triangle $p_i p_{\ell} p_j$, then we have
$|p_i p_j|\geq 1$ and $\sphericalangle (p_{\ell}p_i p_j) \leq \pi/15$, and
$|p_i p_{\ell}| \leq |p_i p_j|+1/10 \leq \frac{11}{10}|p_i p_j|.$

By the law of cosines, we have
\begin{align*}
|p_\ell p_j|^2 &= |p_i p_j|^2 + |p_i p_\ell|^2 - 2|p_i p_j| |p_i p_\ell| \cos(\sphericalangle (p_{\ell}p_i p_j))\\
&=  |p_i p_j|^2 + |p_i p_\ell|\big( |p_i p_\ell| - 2|p_i p_j|\cos(\sphericalangle (p_{\ell}p_i p_j)) \big)\\
&<  |p_i p_j|^2 + |p_i p_\ell|\left( \frac{11}{10}|p_i p_j| - 2\cos\left(\frac{\pi}{15}\right) |p_ip_j| \right)\\
&< |p_i p_j|^2,
\end{align*}
where the first inequality uses our earlier bounds on $|p_i p_{\ell}|$ and $\sphericalangle (p_{\ell}p_i p_j)$,
and the second uses that $\frac{11}{10}-2\cos(\frac{\pi}{15})<0$. This shows that $|p_i p_j|>|p_j p_{\ell}|$ as claimed earlier.

Because the ordering of the pairwise distances in the matrix~$D$ is the same as for the
Euclidean distances, we can conclude that $D[i,j]>D[j,\ell]$ and $D[k,\ell]>D[i,k]$.
But then we can exchange $p_i p_j$ and $p_k p_{\ell}$
for $p_i p_k$ and $p_j p_{\ell}$ in the tour~$T$---because both edges are oriented
from inside $\sig$ to outside $\sig$ this gives a valid tour---and get a shorter tour.
This contradicts the minimality of the tour, concluding the proof of~(PP1).
\medskip

Property (PP2) is a direct consequence of (PP1). Indeed, if we cover $\sigin$ by $O(1)$ hypercubes
of diameter~$\size(\sigin)/5$, then any segment of length at least $\size(\sigin)/4$ inside $\sigin$
crosses at least one such hypercube, and by (PP1) each hypercube is crossed by $O(1)$ edges of length
at least~$\size(\sigin)/4$.
\end{proof}

\added{Notice that the above properties hold for squared Euclidean distances, or more generally, for any distance matrix $D$ where $D[i,j] = f(|p_ip_j|)$ where $f$ is a monotone increasing function. We get the following corollary.
\begin{corollary}
Let $d$ be a fixed integer and let $f:\Reals \rightarrow \Reals$ be a monotone increasing function. Then for a set of $n$ points $P=\{p_1,\dots,p_n\}$ in $\Reals^d$ and distances $D[i,j]=f(|p_ip_j|)$, the shortest tour of $P$ with respect to $D$ (that is, the minimum length Hamiltonian cycle) can be computed in $2^{O(n^{1-1/d})}$ time.
\end{corollary}
}

\begin{remark}
It would be useful for applications if the algorithm could work with a distance matrix
that is a constant distortion of the Euclidean distances, that is, a matrix~$D$ such that
$(1/\alpha)\cdot |p_i p_j|\leq D[i,j] \leq \alpha|p_i p_j|$ for some constant $\alpha \geq 1$.
Unfortunately, while (PP2) holds also in this scenario, (PP1) does not.
\end{remark}

\section{Concluding remarks}
In this paper we described a new geometric separation technique, which
resulted in a faster exact algorithm for \etsp. Together with the lower bound
in~\cite{BergBKMZ20}, this settles the complexity of \etsp assuming
ETH, and up to the constant in the exponent.

We believe that our separation technique can be useful for other problems in
Euclidean geometry as well, and in particular for problems where one wishes to
compute a minimum-length geometric structure that satisfies the Packing
Property. An example of such a problem is {\sc Rectilinear Steiner Tree}. An
additional issue to overcome here is that the number of potential Steiner
points is $O(n^d)$, which means that a direct application of our techniques
does not work. Another challenging problem is finding the minimum weight
triangulation for a set of $n$ points given in $\mathbb{R}^2$, which was
proven \NP-hard by Mulzer and Rote~\cite{MulzerR08} and for which an
$n^{O(\sqrt{n})}$ algorithm is known~\cite{Lingas98}. A minimum-weight
triangulation does not have the Packing Property, because of clusters of
points that are far from each other, but finding an optimal triangulation
between such clusters can perhaps be handled separately.

\bibliography{highdim}

\end{document}